\documentclass[letterpaper,11pt]{article}
\usepackage{amssymb, amsmath, amsthm}
\usepackage{epsfig}
\usepackage{algorithm}
\usepackage{algorithmic}
\usepackage[square,comma]{natbib}
\usepackage{makeidx}
\usepackage{color}
\usepackage[pdftex,colorlinks]{hyperref}
\usepackage{url}
\usepackage{graphicx}
\usepackage{a4}
\newcommand{\comment}[1]{}
\newtheorem{theorem}{Theorem}

\newtheorem{cor}[theorem]{Corollary}

\newtheorem{fact}[theorem]{Fact}
\newtheorem{lemma}[theorem]{Lemma}
\newtheorem{definition}[theorem]{Definition}
\begin{document}
\title{Parallelized approximation algorithms for minimum routing cost spanning
trees}
\author{Ching-Lueh Chang\thanks{Department of Computer Science and
Information Engineering, National Taiwan University, Taipei, Taiwan. Email:
Email: d95007@csie.ntu.edu.tw}
\and Yuh-Dauh Lyuu\thanks{Department of Computer Science and Information
Engineering, National Taiwan University, Taipei, Taiwan. Email:
lyuu@csie.ntu.edu.tw}}
\maketitle

\begin{abstract}
Let $G=(V,E)$ be an undirected graph with a nonnegative edge-weight function
$w.$
The routing cost of a spanning tree $T$ of $G$ is
%the sum over $u,v\in V$ of $d_T(u,v),$
$\sum_{u,v\in V} d_T(u,v),$
where $d_T(u,v)$ denotes the weight of the simple $u$-$v$ path in $T.$
%The problem of finding a spanning tree of approximately minimum routing cost
%is a well-researched problem \cite{Hu74, JLK78, DF79, Won80, WCT00, WLBCRT00}.
The {\sc Minimum Routing Cost Spanning Tree} (MRCT) problem \cite{WLBCRT00} asks for a spanning tree of $G$ with the
minimum routing cost.
In this paper, we parallelize several previously proposed approximation
algorithms for the MRCT problem and some of its variants.
%In this paper, we propose several parallelized approximation
%algorithms for the MRCT problem and some of its variants.
%finding spanning trees of approximately minimum routing cost.
%In particular, we
%parallelize Wu, Chao and Tang's algorithm \cite{WCT00MRCT}, which computes a
%spanning tree with at most $4/3+\epsilon$ times the
%minimum routing cost, where $\epsilon>0$ is an arbitrary constant.
Let $\epsilon>0$ be an arbitrary constant.
%Wu et al. \cite{WCT00MRCT} propose a polynomial-time
%$(4/3+\epsilon)$-approximation algorithm for the MRCT problem.
%That is,
%their algorithm outputs a spanning tree $T$ of $G$ whose routing cost is
%at most $4/3+\epsilon$ times the minimum (over all spanning trees of $G$).
When the edge-weight function $w$ is given in unary,
we parallelize the $(4/3+\epsilon)$-approximation algorithm
for the MRCT problem \cite{WCT00MRCT} by implementing it using an ${\cal
RNC}$ circuit.
%In particular, we
%parallelize Wu, Chao and Tang's $(4/3+\epsilon)$-approximation algorithm
%\cite{WCT00MRCT}.
\comment{%too easy
When $G$ is a complete graph and $w$ obeys the triangle inequality, we
%demonstrate an $\textbf{RNC}^2$
parallelize the $(1+\epsilon)$-approximation algorithm due to \cite{WLBCRT00}, again into
${\cal RNC}$ circuits.
}%too easy
%that computes a spanning tree
%with at most $(1+\epsilon)$ times the minimum routing cost, where
%$\epsilon>0$ is an arbitrary constant.
There are other variants of the MRCT problem.
In the {\sc Sum-Requirement Optimal Communication Spanning Tree} (SROCT) problem
\cite{WCT00PROCTSROCT}, each vertex $u$ is associated with a requirement $r(u)\ge 0.$
The objective is to find a spanning tree $T$ of $G$
minimizing $\sum_{u,v\in V} \left(r(u)+r(v)\right) \, d_T(u,v).$
%the sum over $u,v\in V$ of $r(u)+r(v)$ times $d_T(u,v).$
\comment{%too easy
The {\sc Product-Requirement Optimal Communication Spanning Tree} (PROCT)
\cite{WCT00PROCTSROCT} problem is similar except that the objective is to minimize
$\sum_{u,v\in V} r(u)\, r(v) \, d_T(u,v).$
}%too easy
%the sum over $u,v\in V$ of $r(u) \, r(v)$ times $d_T(u,v).$
%Wu et al. \cite{WCT00PROCTSROCT} propose a $2$-approximation algorithm for
%the SROCT problem.
%When the edge-weight function $w$ is strictly positive,
When the edge-weight function $w$ and the vertex-requirement function $r$
are given in unary,
we parallelize the $2$-approximation algorithm for the SROCT problem
\cite{WCT00PROCTSROCT} by realizing it using
${\cal RNC}$ circuits, with a slight degradation in the approximation ratio
from $2$ to $2+o(1).$
%Wu et al. \cite{WCT00PROCTSROCT} also propose a $1.577$-approximation
%algorithm for the PROCT problem.
\comment{%too easy
When
$w$ and $r$ are integer-valued,
bounded by a polynomial in $|V|$ and $w$
obeys the triangle inequality,
we parallelize the $1.577$-approximation algorithm for the PROCT problem
\cite{WCT00PROCTSROCT} by realizing it using
${\cal RNC}$ circuits.
}%too easy
In the weighted $2$-MRCT problem \cite{Wu02}, we have additional inputs
$s_1,s_2\in V$ and $\lambda\ge 1.$
%two source vertices
%$s_1,s_2\in V$ and $\lambda\ge 1$ are given.
The objective is to find a spanning tree $T$ of $G$ minimizing
$\sum_{v\in V} \lambda \, d_T(s_1,v) + d_T(s_2,v).$
%$\lambda$
%times the sum over $v\in V$ of the weight of the simple $s_1$-$v$ path in
%$T,$ plus the sum over $v\in V$ of the weight of the simple $s_2$-$v$ path
%in $T.$
%Wu \cite{Wu02} shows a $2$-approximation algorithm for this problem and
When the edge-weight function $w$ is given in unary, we
parallelize the $2$-approximation algorithm \cite{Wu02} into ${\cal RNC}$ circuits, with a slight
degradation in the approximation ratio from $2$ to $2+o(1).$
%All our parallelized algorithms are realized by \textbf{RNC} circuits.
%Our results are primarily based on \cite{RA00, WCT00, WLBCRT00, WCT00b, WCT00c}.
%Our results are primarily based on Reinhardt and Allender \cite{RA00}
%and Wu, Chao and Tang \cite{WCT00MRCT} and Wu et al. \cite{WLBCRT00}.
To the best of our knowledge, our results are the first
parallelized approximation algorithms for the MRCT problem and its
variants.
\end{abstract}

\section{Introduction}
Let $G=(V,E)$ be an undirected graph with a nonnegative edge-weight function
$w.$
The routing cost of a spanning tree $T$ of $G$
%is the sum over $u,v\in V$ of
is $\sum_{u,v \in V} d_T(u,v)$ where $d_T(u,v)$ is
the weight of any shortest $u$-$v$ path in $T,$
or equivalently, the weight of the simple $u$-$v$ path in $T.$
%Given $G,w,$
The {\sc Minimum Routing Cost Spanning Tree} (MRCT) problem
\cite{WLBCRT00} asks for a
spanning tree $T$ of $G$ with the minimum routing cost.
It is also known as the {\sc Shortest Total Path Length Spanning Tree} problem.
%In some papers, the MRCT problem is named the shortest total path length
%spanning tree problem.
The MRCT problem is first proposed by Hu \cite{Hu74}, who referred to the
problem as the {\sc Optimum Distance Spanning Tree Problem}.
In Hu's formulation of the more general
{\sc Optimum Communication Spanning Tree} (OCT) problem \cite{Hu74}, an additional value
$\tau_{u,v}\ge 0$ is given for each pair $(u,v)$ of vertices.
%For a graph $G=(V,E)$ with a nonnegative edge-weight function $w,$
The communication cost \cite{Hu74} of a spanning tree $T$ of $G$ is
$\sum_{u,v\in V} \tau_{u,v} \, d_T(u,v).$
%where $d_T(u,v)$ is the weight of any
%shortest $u$-$v$ path in $T.$
%Since the edge-weight function is nonnegative,
%$d_T(u,v)$ equals the weight of the simple $u$-$v$ path in $T.$
%the sum over pairs
%$(u,v)\in V^2$ of $r_{u,v}$ multiplied by the weight of the simple $u$-$v$
%path in $T.$
The OCT problem asks for a spanning tree of $G$ with the minimum
communication cost.
When $G$ is a complete graph and the edge-weight function $w$ obeys the triangle inequality, a randomized
$O(\log {|V|})$-approximation algorithm is known for the OCT problem \cite{Bar98, CCGG98, WLBCRT00, FRT03}.
The MRCT problem is the special case of the OCT problem when $\tau_{u,v}=1$ for all $u,v\in V.$

The MRCT problem has applications in network design \cite{Hu74,JLK78}
as well as multiple sequences alignment in computational biology \cite{FD87,
Pev92, Gus93, BLP94, WLBCRT00}.
Unfortunately, it is shown to be ${\cal NP}$-hard \cite{JLK78}, and
it is ${\cal NP}$-hard even when all edge weights are equal \cite{JLK78,GJ79}
or when the edge-weight function obeys the triangle inequality \cite{WLBCRT00}.

Exact and approximation algorithms for the MRCT problem have been extensively
researched \cite{BFW73, Hoa73, DF79, Won80, WCT00MRCT, WLBCRT00, FLS02}.
Boyce et al. \cite{BFW73}, Hoang \cite{Hoa73} and Dionne and
Florian \cite{DF79} study branch-and-bound algorithms as well as heuristic
approximation algorithms for the {\sc Optimal Network Design} problem \cite{BFW73},
which includes the MRCT problem as a special case.
Fischetti et al. \cite{FLS02} give exact algorithms for the MRCT problem
while avoiding exhaustive search.
%Dionne and Florian \cite{DF79} give exact and heuristic approximation
%algorithms for the MRCT problem.
Wong \cite{Won80} gives a polynomial-time $2$-approximation algorithm for the MRCT problem.
That is, he gives a polynomial-time algorithm that, given a graph $G=(V,E)$ with
a nonnegative edge-weight function $w,$ outputs a spanning tree of $G$ whose
routing cost is at most $2$ times the minimum.
Subsequent work by Wu et al. \cite{WCT00MRCT} shows a different
polynomial-time $2$-approximation algorithm as well as polynomial-time
$15/8, 3/2$ and $(4/3+\epsilon)$-approximation algorithms for the MRCT
problem, where $\epsilon>0$ is an arbitrary constant.
Their results are later improved by Wu et al. \cite{WLBCRT00} to give a
polynomial-time approximation scheme (PTAS) \cite{CLRS01} for the MRCT problem.
%A PTAS \cite{CLRS01} for a minimization problem
%is an algorithm that, given an additional $\epsilon>0$ besides the ordinary
%input, outputs a feasible solution whose objective function is at most
%$(1+\epsilon)$ times the minimum (over all feasible solutions).
%The running time of the PTAS has to be polynomial in its input size for
%each constant $\epsilon>0.$
%For example,
%in the case for the MRCT problem, a PTAS outputs a spanning tree $T$ of the
%input graph $G=(V,E),$ and the value of $\sum_{u,v\in V}
%d_T(u,v)$ is at most $(1+\epsilon)$ times the minimum (over all spanning
%trees of $G$). The running time of the PTAS is polynomial in its input size
%for each $\epsilon>0.$
That is, a polynomial-time $(1+\epsilon)$-approximation algorithm is given for
any constant $\epsilon>0.$

There are other variants of the MRCT problem that also have
applications in network design \cite{WLBCRT00, WCT00PROCTSROCT,
WCT00PROCTPTAS, Wu02}.
In the {\sc Sum-Requirement Optimal Communication Spanning Tree} (SROCT) problem
\cite{WCT00PROCTSROCT}, each vertex $u$ is associated with a requirement $r(u)\ge 0.$
The objective is to find a spanning tree $T$ of $G$
minimizing $\sum_{u,v\in V} (r(u)+r(v)) \, d_T(u,v).$
%the sum over $u,v\in V$ of $r(u)+r(v)$ times $d_T(u,v).$
The {\sc Product-Requirement Optimal Communication Spanning Tree} (PROCT)
\cite{WCT00PROCTSROCT} problem is to find a spanning tree $T$ of $G$ minimizing
$\sum_{u,v\in V} r(u)\, r(v) \, d_T(u,v).$
%the sum over $u,v\in V$ of $r(u) \, r(v)$ times $d_T(u,v).$
The SROCT and PROCT problems are clearly generalizations of the MRCT problem.
%and are
%thus \textbf{NP}-hard \cite{JLK78} even when all edge weights are equal
%\cite{JLK78, GJ79} or when the edge-weight function
%obeys the triangle inequality \cite{WLBCRT00}.

Wu et al. \cite{WCT00PROCTSROCT} give a $2$-approximation algorithm for
the SROCT problem.
%To our knowledge this is the currently best approximation ratio.
%Wu et al. \cite{WCT00PROCTSROCT}
They also propose a $1.577$-approximation
algorithm \cite{WCT00PROCTSROCT} for the PROCT problem.
The result is improved by Wu et al. \cite{WCT00PROCTPTAS} to yield a
polynomial-time approximation scheme (PTAS) for the PROCT problem.
%A polynomial-time approximation scheme (PTAS) \cite{CLRS01} for a minimization
%problem
%is an algorithm that, given an additional $\epsilon>0$ besides the ordinary
%input, outputs a feasible solution whose objective function is at most
%$(1+\epsilon)$ times the minimum (over all feasible solutions).
%The running time of the algorithm has to be polynomial in its input size for
%each $\epsilon>0.$
%In the case for the PROCT problem, a PTAS outputs a spanning tree $T$ of the
%input graph $G=(V,E),$ and the value of $\sum_{u,v\in V} r(u) \, r(v) \,
%d_T(u,v)$ is at most $(1+\epsilon)$ times the minimum (over all spanning
%trees of $G$). The running time of the PTAS is polynomial in its input size
%for each $\epsilon>0.$

Another variant of the MRCT problem is the $2$-MRCT problem
\cite{Wu02}.
In this problem, except for $G=(V,E)$ and $w:E\to \mathbb{R}_0^+,$ we are
given two source vertices
$s_1,s_2\in V.$
The objective is to find a spanning tree $T$ of $G$ minimizing
$\sum_{v\in V} d_T(s_1,v) + d_T(s_2,v).$
%(over all spanning trees of $G$).
This problem is ${\cal NP}$-hard even when $w$ obeys the
triangle inequality \cite{Wu02}.
Wu \cite{Wu02} shows a $2$-approximation algorithm as well as
a PTAS for this problem.
A variant of the $2$-MRCT problem is the weighted $2$-MRCT problem
\cite{Wu02} where an additional $\lambda\ge 1$ is given as input.
The objective is to find a spanning tree $T$ of $G$ minimizing
$\sum_{v\in V} \lambda \, d_T(s_1,v) + d_T(s_2,v).$
%(over all spanning trees of $G$).
Wu \cite{Wu02} proposes a $2$-approximation algorithm for the weighted
$2$-MRCT problem.
%To our knowledge this is the currently best approximation ratio.
When the edge-weight function $w$ obeys the
triangle inequality, there is a PTAS for the weighted $2$-MRCT problem \cite{Wu02}.

In this paper, however, we will focus on parallelizing the approximation
algorithms for the above problems.
We first describe our results concerning the MRCT problem.
For an arbitrary $\epsilon>0$ and when the edge-weight function $w$ is given
in unary,
we show that the
$(4/3+\epsilon)$-approximation algorithm proposed by Wu et al.
\cite{WCT00MRCT} can be implemented by an ${\cal RNC}$ circuit.
%when the edge-weight function $w$ is given in unary.
That is, the approximation algorithm can be performed by a uniform
polynomial-size circuit \cite{Pap94CC} with random gates and poly-logarithmic
depth.
Indeed, with a small probability our parallelized algorithm may fail to find a
$(4/3+\epsilon)$-approximate solution, in which case it outputs ``fail.''
Thus, our algorithm does not fail (to find a
$(4/3+\epsilon)$-approximate solution) without ever knowing that it fails,
which is a desirable property for randomized algorithms with a small
probability of failure.
\comment{%too easy
We also show that when the edge-weight function $w$ is integer-valued,
bounded by a polynomial in $|V|,$ and obeys the triangle
inequality,
the $(1+\epsilon)$-approximation algorithm proposed by
Wu et al. \cite{WLBCRT00} can be realized by an ${\cal RNC}$ circuit.
}%too easy

We now turn to describe our results concerning the SROCT problem.
%When the edge-weight function $w$ is strictly positive,
When the edge-weight and the vertex-requirement functions are given in
unary, we parallelize the $2$-approximation
algorithm \cite{WCT00PROCTSROCT} by realizing it using
${\cal RNC}$ circuits, with a slight degradation in the
approximation ratio (from the currently best $2$ to our $2+o(1)$).
Still, with a small probability
our algorithm may fail to output a $(2+o(1))$-approximate solution,
in which case it knows the failure and outputs ``fail.''

\comment{%too easy
For the PROCT problem,
when the edge-weight function $w$
%$w:E\to \{0,\ldots,\text{poly}(|V|)\}$
and
the vertex-requirement function $r$
%$r:E\to \{0,\ldots,\text{poly}(|V|)\}$
are integer-valued, bounded by a polynomial in $|V|$ and $w$
obeys the triangle inequality,
we parallelize the $1.577$-approximation algorithm \cite{WCT00PROCTSROCT}
by realizing it using ${\cal RNC}$ circuits.
}%too easy

Finally, for the weighted $2$-MRCT problem with the edge-weight function
given in unary,
we parallelize the $2$-approximation algorithm \cite{Wu02} into
${\cal RNC}$ circuits, with a slight
degradation in the approximation ratio (from the currently best $2$ to our $2+o(1)$).
Again, there is a small probability that our algorithm fails to find a
$(2+o(1))$-approximate solution, in which case it outputs ``fail.''

To the best of our knowledge, our results are the first efforts towards
parallelized approximation algorithms for the MRCT problem and its
variants.
Our results open up new opportunities to compute approximate solutions to
the above problems in
parallel poly-logarithmic time.
In the applications of the MRCT problem to network design \cite{Hu74,
JLK78} as well as the applications of the SROCT and PROCT problems to network
design \cite{WCT00PROCTSROCT, WCT00PROCTPTAS}, the network is often modeled as a graph with a nonnegative
edge-weight function representing the distances between pairs of nodes.
%Johnson, Lenstra and Kan \cite{JLK78} uses graphs to model networks where
%the edge weights represent the communication costs between nodes, and they
%show that the MRCT problem is \textbf{NP}-hard.
%In \cite{JLK78}, the MRCT problem is used to model networks where the edge
%weights represent the communication costs between vertices.
Although approximate solutions to the aforementioned problems (MRCT, SROCT,
PROCT, weighted $2$-MRCT) are attainable in polynomial time,
in any real networking environment, however, the cost of traffic
between any pair of nodes may vary over time.
Thus, it is highly desirable to be able to compute approximate solutions to
these problems as fast as possible, so as to reduce the risk that the
traffic costs change during the computation.
Our results imply that approximate solutions to the MRCT,
SROCT and weighted $2$-MRCT problems can indeed
by computed in parallel poly-logarithmic time on multiprocessors.

For other applications of the MRCT problem where the data does not change
quickly over time, for example multiple sequences alignment in
computational biology \cite{FD87, Pev92, Gus93, BLP94, WLBCRT00},
being able to compute approximate solutions to the
MRCT problem in parallel sublinear time is still beneficial.
Indeed, Fischer \cite{Fis01} argues that in many practical applications
today, the input
size is so large that even performing linear-time computations is too time-consuming.
Certainly, multiple sequences alignment in computational biology constitutes
a good example where the input size is usually too large.
It is therefore a desirable property that our algorithms operate in
parallel sublinear time, and in fact poly-logarithmic time.
%On the other hand,
%Our parallelized algorithms operate in parallel
%sublinear time, and in fact poly-logarithmic time.
%Our results deal with the problem in a different perspective.
%The property testing methods only do not read the entire input.
%Instead, our parallelized approximation algorithms read the entire input but
%performs computation in square-logarithmic time, which is certainly
%sublinear time.

The main idea underlying our proofs is that many of the previously proposed
approximation algorithms for the MRCT, SROCT and weighted $2$-MRCT problems
rely heavily on finding shortest paths between pairs of vertices in a
graph. This motivates applying the well-known result that ${\cal
NL}\subseteq {\cal NC}$ to parallelize these algorithms since we can guess a
path (possibly the shortest one) between two vertices of a graph in nondeterministic logarithmic space.
There is the complication that, in our proofs, we will often need to
generate the same shortest path between two vertices $u,v,$ whenever
a shortest $u$-$v$ path is needed. For this purpose, we use the isolation technique \cite{Wig94,
GW96, RA00} to slightly modify the edge-weight function of the input graph,
so that there is exactly one shortest path between each pair of vertices
with high probability. We then apply the double-counting technique
\cite{RA00} to decide whether the input graph (with the modified
edge-weight function) exhibits a unique shortest path between each pair $u,v$
of vertices. If so, we are able use the double counting technique
to generate the unique shortest $u$-$v$ path whenever it is needed.
The whole
procedure runs in unambiguous logarithmic space and our results follow
by ${\cal UL}\subseteq {\cal NL}\subseteq {\cal NC}.$ The approximation ratio would be slightly
degraded. The degradation comes from randomly modifying the edge-weight function
when we apply the isolation technique.

Our paper is organized as follows. Section~\ref{defs} provides the basic
definitions. Section~\ref{firstresult} presents the parallelized
$(4/3+\epsilon)$-approximation algorithm for the MRCT problem.
\comment{%too easy
Section~\ref{secondresult} presents the parallelized
$(1+\epsilon)$-approximation algorithm for the MRCT problem when the
edge-weight function obeys the triangle inequality.
}%too easy
Section~\ref{SROCTsection}--\ref{twoMRCTsection}
describe our parallelized approximation algorithms
for the SROCT and the weighted $2$-MRCT problems, respectively.
Section~\ref{conclusion} concludes the paper.
Proofs are given in the appendix for references.

\section{Notations and basic facts} \label{defs}
Throughout this paper, graphs are simple undirected graphs \cite{Wes01}.
That is, we disallow parallel edges and self-loops.
There will always be a nonnegative edge-weight function mapping each edge to
a nonnegative real number.
For a graph $G,$ $V(G)$
is its vertex set and $E(G)$ is its edge set. Let $R$ be a subgraph of $G.$
A path $P$ connects a vertex $v$ to $R$ (or $V(R)$) if one endpoint of $P$ is $v$ and the other is in
$V(R).$
An edge connecting two vertices $u$ and $v$ is denoted $uv.$
A path connecting two vertices $u$ and $v$ is said to be a $u$-$v$ path.
A path $(v_0,\ldots,v_k)$ is one which traverses $v_0,\ldots,v_k,$
in that order.
A simple path is a path that traverses each vertex at most once \cite{Wes01}.
A graph $G$ contains another graph $G^\prime$ if $G^\prime$ is a subgraph of
$G.$
The set of nonnegative real numbers is denoted $\mathbb{R}_0^+.$
\begin{definition} \label{basicdef}
Let $G=(V,E)$ be an undirected graph and $w:E\to \mathbb{R}^+_0$
be a nonnegative edge-weight function.
The lexicographical ordering on $V$ is that of the encodings of the vertices in
$V,$ assuming any reasonable encoding of a graph.
Let $u,v\in V$ and $R$ be a subgraph of $G.$
The sum of edge weights of $R$ is denoted $w(R).$
When $R$ is a path, $w(R)$ is called the weight or length of $R.$
%The weight or length of a path $P$ is $w(P).$
For $x,y\in V,$
we use $d_G(x,y)$ to denote the weight of any shortest $x$-$y$ path.
%$d_G(x,R)$ (or $d_G(x,V(R))$) denotes the weight of any shortest path connecting $x$ and
%$V(R).$
We use $d_G(x,R)$ (or $d_G(x,V(R))$) for $\min_{z\in V(R)} d_G(x,z).$
%When $V(R)=\{y\}$ is a singleton, $d_G(x,R)$ (or $d_G(x,V(R))$) is also denoted
%$d_G(x,y).$
The lexicographically first vertex $x^\prime\in V(R)$
satisfying $d_G(x,x^\prime)=d_G(x,R)$
is denoted $\text{closest}(x,R)$ (or $\text{closest}(x,V(R))$).
The set of all shortest paths connecting $u$ and $v$ is denoted
$\text{SP}_G(u,v).$
%$\text{SP}_G(u,R)$ denotes the set of shortest paths connecting $v$ and
%$V(R).$
$\text{SP}_G(u,R)$ (or $\text{SP}_G(u,V(R))$) is the set of shortest paths connecting $u$ and $R.$
That is, $\text{SP}_G(u,R)$ is the set of paths
that connect $u$ and $R$ and have weight equal to $d_G(u,R).$
%that is
%one of the shortest paths with one endpoint being $u$ and the
%other in $R.$
%The vertex set of $G$ is also denoted $V(G),$
%and its edge set is also denoted $E(G).$
%The unit edge-weight is also denoted $\boldmath{1}.$
For $k\in \mathbb{N},$ $S_{k,u}$ denotes the set of vertices reachable from
$u$ with at most $k$ edges.
$\text{SP}^{(k)}_{u,v}$ denotes the set of
all shortest paths among those $u$-$v$ paths with at most $k$ edges.
That is, $\text{SP}^{(k)}_{u,v}$ is the set of $u$-$v$ paths with at
most $k$ edges whose weight is not larger than any other path with at most $k$ edges.
%For any subgraph $Y$ of $G,$ its sum of edge weights is denoted $w(Y).$
%Let $S$ be a subtree of $G.$
The union of two graphs $G_1=(V_1,E_1)$ and $G_2=(V_2,E_2)$
%with the same vertex set $V$
is the graph $(V_1\cup V_2,E_1\cup E_2).$
The graph $G$
is strongly min-unique with respect to $w$ if
for all $k\in \mathbb{N}$ and $u,v\in V,$ we have
$\left|\text{SP}^{(k)}_{u,v}\right|\le 1.$
When $w$ is clear from the context, we may simply say that $G$ is strongly
min-unique without referring to $w.$
%every pair of vertices of $G$ exhibits a unique
%shortest path between them.
\end{definition}
In Definition~\ref{basicdef},
it is not hard to show that $G=(V,E)$ is strongly min-unique if
$\left|\text{SP}^{(k)}_{u,v}\right|\le 1$ for all $k\in \{0,\ldots,|V|-1\}$ and
$u,v\in V,$ provided $|V|\ge 3.$
%$|V|\ge 3$
%and for all $0\le k\le |V|-1$ and $u,v\in V,$
%we have $\left|\text{SP}^{(k)}_{u,v}\right|\le 1.$
%The definition of strong min-uniqueness above is due to Reinhardt and Allender,
%but they define it over directed graphs.
%the weight function $w$ will be extended to map a subgraph of
%$G$ to the sum of its edge weights.

The MRCT of a graph, standing for its \textbf{M}inimum \textbf{R}outing
\textbf{C}ost spanning \textbf{T}ree, is defined below.
\begin{definition} (\cite{CW04}) \label{routingcost}
Given a connected graph $G=(V,E)$ with a nonnegative edge-weight
function $w:E\to \mathbb{R}^+_0,$
the routing cost $c_w(T)$ of
a spanning tree $T$ of $G$ is
$\sum_{u,v\in V} d_T(u,v).$
%where $P\in \text{SP}_T(u,v).$
A spanning tree of $G$ with the minimum routing cost is an MRCT of $G,$
which is denoted by $\text{MRCT}(G)$ for convenience.
%The unit edge-weight is also denote $\boldmath{1}.$\footnote{Ching-Lueh: Or use bold
%font for \bf{1}?\\
%Lyuu: no need unless your 1 is a vector\\
%Ching-Lueh: What if it is a function?}
\end{definition}
The MRCT problem asks for $\text{MRCT}(G)$ on input $G,w.$
The following fact shows that the routing cost of a tree can be computed
efficiently.
\begin{fact} (\cite{CW04}) \label{loading}
%Let $G=(V,E)$ be a connected, weighted and undirected graph with nonnegative edge-weight
%$w:E\to \mathbb{R}^+_0$
%and $T$ be a spanning tree of $G.$
Let $G$ be a graph with a nonnegative edge-weight function $w:E(T)\to \mathbb{R}^+_0$
and $T$ be a spanning tree of $G.$
For each edge $e\in E(T),$ let $T_{e,1}$ and $T_{e,2}$ be the two trees formed by
removing $e$ from $T.$
%and $\ell(e,T)\stackrel{def}{=} 2|T_{e,1}|\, |T_{e,2}|.$
We have $c_w(T)=\sum_{e\in E(T)} 2|V(T_{e,1})|\, |V(T_{e,2})| \, w(e)$
and
%In particular, $\ell(e,T)\le |V(T)|^2/2$ for each $e\in E(T)$ and thus
$c_w(T)\le |V(T)|^3/2 \cdot \max_{e\in E(T)} w(e).$
\end{fact}
To ease the description, we introduce the following definition.
\begin{definition} (\cite{RA00}) \label{unambdef}
A nondeterministic Turing machine $M$ outputs a string $s$ unambiguously
on input $x$ if it outputs $s$ on exactly one non-rejecting computation branch,
and rejects $x$ on all other computation branches.
The unambiguously output string $s$ is also denoted $M(x).$
\end{definition}
Throughout this paper, when a nondeterministic Turing
machine $A$ runs or simulates another nondeterministic machine $B,$ it means that $A$
runs $B$ and make nondeterministic branches as $B$ does. It does not mean
that $A$ enumerates all computation branches of $B$ and simulate them deterministically.
For convenience, $A$ does not necessarily have to output $B$'s output.
Instead, it may extract portions of $B$'s output for output.

We will need the notion of a general star to introduce the
approximation algorithms for the MRCT problem.
\begin{definition} (\cite{CW04})
Let $G$ be a connected graph with a nonnegative edge-weight function $w: E\to
\mathbb{R}^+_0$ and $S$ be a subtree of $G.$
A spanning tree $T$ containing $S$ is a general star with core $S$
if
%, in the tree $T,$
each vertex $u\in V$
satisfies $d_T(u,S)=d_G(u,S).$
%is connected to $S$
%using one of the paths in $\text{SP}_G(u,S).$
When $V(S)=\{v\}$ is a singleton,
a general star with core $S$ is also
called a shortest path tree rooted at $v.$
%the shortest paths in $G.$
\end{definition}
Given any subtree $S$ of $G=(V,E),$ a general star with core $S$ exists
\cite{WCT00MRCT}.
This follows by observing that for any shortest path $P$ connecting $u\in V$ and
$S,$ the part of $P$ from any vertex $x\in V(P)$ to $S$ constitutes a
shortest path connecting $x$ and $S.$

The notion of a metric graph is defined below.
\begin{definition} (\cite{WLBCRT00})
A complete graph $G$ with a non-negative edge-weight
function $w$ is metric if $w(xy)+w(yz)\ge w(xz)$ for all $x,y,z\in V(G).$
\end{definition}

\section{A parallelized $(4/3+\epsilon)$-approximation for MRCT} \label{firstresult}
We begin with
the following form of the famous isolation lemma. It is implicit in some
previous works \cite{Wig94, GW96, RA00}.
%We include a proof in the appendix.
%For completeness we include a proof to it.
\begin{theorem} (\cite{RA00}) \label{isolation}
%Let $X_{i,j}, 1\le i,j\le n$ be independent random variables with uniform
%distribution over $\{1,\ldots,n^5\}$ where $n\in\mathbb{N}.$
%With high probability every weighted undirected $n$-node graph
%$G=\left(\{1,\ldots,n\},E\right)$ is min-unique,
%where the weight of each edge $e=h\rightarrow k\in E$ is given by $X_{h,k}.$
Let $G=(V,E)$ be a graph with a nonnegative edge-weight
$w: E\to \mathbb{R}^+_0.$
Let $w_r: E\to \mathbb{R}^+_0$ assign the weight of each $e\in E$
independently and randomly from the uniform distribution over a set
$W\subseteq \mathbb{R}^+_0.$
%$\{1,\ldots,{|V|}^5\}.$
With probability at least $1-{|V|^5}/{(2|W|)},$ the graph $G$ is strongly min-unique with
respect to $w+w_r.$
\end{theorem}
%The following lemma explores
%the relation between random edge-weight versus unit edge-weight on the routing cost
%of trees.
%In this version we omit the proof.
%The following lemma is stated without proof in this version.
%\begin{lemma} \label{randomcost}
%Let $T$ be a tree with vertex set $V.$
%Assign the weight of each edge of $T$ independently and randomly from the
%uniform distribution over $\{1,\ldots,{|V|}^5\}.$
%With probability at least $1-\exp{(-|V|)}$,
%$$1-\frac{1}{|V|}\le\frac{c_w(T)}{(1+|V|)^5/2 \cdot c_{\boldmath{1}}(T)}\le
%1+\frac{1}{|V|}.$$
%\end{lemma}

The following theorem is implicit in \cite{RA00}.
It uses the double counting technique \cite{RA00} similar to the inductive
counting technique used to prove the Immerman-Szelepcs{\'e}nyi theorem
\cite{Imm88, Sze88}.
%We include a proof in the appendix.
%Reinhardt and Allender have the following theorem \cite{RA00}.
\begin{theorem} (\cite{RA00}) \label{doublecount}
There is a nondeterministic logarithmic-space Turing machine
FIND-PATH that,
on input a graph $G=(V,E)$ with a nonnegative edge-weight function
$w:E\to \{0,\ldots,\text{poly}(|V|)\}$ and two vertices $s,t \in V,$
satisfies the following conditions.
\begin{itemize}
\item [1.] If $G$ is not strongly min-unique, then FIND-PATH
outputs ``not strongly min-unique'' unambiguously.
\item [2.] If $G$ is strongly min-unique and has an $s$-$t$ path, then FIND-PATH
outputs the unique path $P\in\text{SP}_G(s,t)$ and its weight $w(P)$ unambiguously.
The edges in $P$ are output in the direction going from $s$ to $t.$
\item [3.] If $G$ is strongly min-unique and does not have an $s$-$t$ path, then FIND-PATH has no
accepting computation branches.
\end{itemize}
\end{theorem}

The following theorem is due to Wu et al. \cite{WCT00MRCT}.
\begin{theorem}(\cite{WCT00MRCT})\label{approx}
Let $r\in\mathbb{N}$ be a constant
and $G=(V,E)$ be a connected, strongly min-unique graph
with a nonnegative edge-weight function $w: E\to \mathbb{R}^+_0.$
For $1\le k\le r+4$ and $S=(v_1,\ldots,v_k)\in V^k,$
let $R_{1,S}$ be the subgraph of $G$ containing only $v_1.$
For $2\le i\le k,$
let %$R_{i,S}=R_{i-1,S}\cup P\left(v_i,V(R_{i-1,S})\right)$
$R_{i,S}=R_{i-1,S}\cup P_{i,S}$
where $P_{i,S}\in \text{SP}_G\left(v_i,\text{closest}(v_i,R_{i-1,S}))\right)$
is the unique shortest path
connecting $v_i$ and $\text{closest}(v_i,R_{i-1,S}).$
%and visiting $R$ only once.
For some $1\le k\le r+4$ and $S\in V^k,$ every general star $T$ with core $R_{k,S}$
satisfies $$c_w(T)\le \left(\frac{4}{3}+\frac{8}{9r+2}\right) \,
c_w(\text{MRCT}(G)).$$
\end{theorem}
That $R_{i,S}$ in Theorem~\ref{approx} is a subtree of $G$ for $2\le i\le
k$ is easily shown because $w(e)>0$ for each $e\in E$ by the
strong min-uniqueness of $G.$
\comment{ % begin removal
The following theorem is due to Chao and Wu \cite{CW04}.
\begin{theorem} (\cite{CW04}) \label{approx}
Let $G=(V,E)$ be a connected, weighted undirected graph
with a nonnegative edge-weight function $w: E\to \mathbb{R}^+_0.$
There exist vertices $x,y,z\in V$ such that for any $P_{x,y}\in
\text{SP}_G(x,y), P_{y,z}\in \text{SP}_G(y,z)$ and any spanning tree $S$ of
$P_1\cup P_2,$ every general star $T$ with core
$S$ satisfies
$$c_w(T)< \frac{3}{2} \cdot c_w(\text{MRCT}(G)).$$
\end{theorem}
} % end removal
The core $R_{k,S}$ in Theorem~\ref{approx} is unambiguously computable in logarithmic space
on strongly min-unique connected graphs.
To show this, we need the following lemma.
%Its proof is included in the appendix.
\begin{lemma} \label{add}
There is a nondeterministic logarithmic-space Turing machine $\text{ADD-PATH}$ that,
on input a strongly min-unique connected graph $G=(V,E)$ with a nonnegative edge-weight
function $w:E\to \{0,\ldots,\text{poly}(|V|)\},$ a subgraph $R$ of $G$ and a vertex $v\in V,$
outputs the unique path $P\in \text{SP}_G\left(v,\text{closest}(v,R)\right)$ unambiguously.
%And $P$ visits $R$ exactly once.
%with $\left|V(P)\cap V(R)\right|=1.$
\end{lemma}
With Lemma~\ref{add}, we are able to compute
the core $R_{k,S}$ in Theorem~\ref{approx} unambiguously
in logarithmic space on strongly min-unique connected graphs.
%The proof is included in the appendix.
\begin{lemma}\label{core}
Let $r\in\mathbb{N}$ be a constant.
There is a nondeterministic logarithmic-space Turing machine CORE that,
on input a strongly min-unique connected graph $G=(V,E)$ with a nonnegative edge-weight
function $w:E\to \{0,\ldots,\text{poly}(|V|)\}$ and $S=(v_1,\ldots,v_k)\in V^k$ where $1\le
k\le r+4,$ unambiguously outputs $R_{k,S}$ defined below.
$R_{1,S}$ is the subgraph of $G$ containing only $v_1.$
For $2\le i\le k,$
$R_{i,S}=R_{i-1,S}\cup P_{i,S}$ where
$$P_{i,S}\in
\text{SP}_G\left(v_i,\text{closest}(v_i,R_{i-1,S})\right)$$ is the unique shortest path
connecting $v_i$ and $\text{closest}(v_i,R_{i-1,S}).$
%Also, $P_{i,S}$ visits $R_{i-1,S}$ exactly once, for $2\le i\le k.$
\end{lemma}
\comment{ % begin removal
\begin{theorem} \label{core}
There is a nondeterministic logarithmic-space Turing machine CORE that, on
input a strongly min-unique, connected, weighted undirected graph $G=(V,E)$
with a nonnegative edge-weight function $w: E\to \mathbb{R}^+_0$ and $x,y,z\in V,$
has the following property.
Let $P_{x,y}\in \text{SP}_G(x,y)$ and $P_{y,z}\in \text{SP}_G(y,z)$ be the unique shortest
paths connecting $x,y$ and $y,z,$ respectively.
CORE computes a spanning tree of $P_{x,y}\cup P_{y,z}$ unambiguously.
\end{theorem}
\begin{proof}
CORE first runs $\text{FIND-PATH}(G,w,x,y)$ to output $P_{x,y}$
unambiguously.
It then simulates $\text{FIND-PATH}(G,w,y,z)$ for $P_{y,z},$
with the following additional check before outputting an edge.
Let $u$ be the most recently generated vertex by
$\text{FIND-PATH}(G,w,y,z)$
and $e=uv$ be the edge that $\text{FIND-PATH}(G,w,y,z)$ is going to output.
CORE simulates $\text{FIND-PATH}(G,w,x,y)$ again to test whether
$v\in V(P_{x,y})$ and outputs $e$ if not so.
%CORE outputs $e$ only if $v\notin V(P_{x,y}).$
If any of the calls to FIND-PATH ends up rejecting, CORE also rejects.
This guarantees that CORE has exactly one non-rejecting computation branch.
It is not hard to see that the output of CORE is a spanning tree of $P_{x,y}\cup P_{y,z}$
\end{proof}
} % end removal
\comment{ %Begin of removal
If we further assume that $G$ is strongly min-unique and $u,v\in V$... a unique path
$P\in \text{SP}_G(u,v).$
For vertices $x,y\notin V(P),$ let $x^\prime, y^\prime\in V(P)$ be the
lexicographically first vertices in $V(P)$ such that
$w\left(\text{SP}_G(x,x^\prime)\right)=w\left(\text{SP}_G(x,P)\right)$
and
$w\left(\text{SP}_G(y,y^\prime)\right)=w\left(\text{SP}_G(y,P)\right),$
respectively.
It is easy to see that if the unique paths $P_x\in
\text{SP}_G(x,x^\prime)$ and $P_y\in \text{SP}_G(y,y^\prime)$ have a node
$z$ in common, then $x^\prime=y^\prime$ and
$P_x$ and $P_y$ follow exactly the same path from $z$ to $x^\prime.$
This observation and Theorem~\ref{approx} give the following corollary.
} % end of removal
%The following fact is now easy to show.
With Theorem~\ref{approx} and Lemma~\ref{core}, it is not hard to show the
following fact.
\begin{fact} \label{substruct}
Let $r\in\mathbb{N}$ be a constant and
$G=(V,E)$ be a strongly min-unique, connected graph with a nonnegative
edge-weight function $w: E\to \{0,\ldots,\text{poly}(|V|)\}.$
For a sequence $S$ of
%Let $S$ be a sequence of
at most $r+4$ vertices in $V,$
let $C_S=\text{CORE}(G,w,S)$
and $P_u\in \text{SP}_G\left(u,\text{closest}(u,C_S)\right)$ for $u\in V\setminus V(C_S).$
Then
$$T_S=C_S \cup \bigcup_{u\in V\setminus V(C_S)} P_u$$
is a general star with core $S,$
and $$c_w(T_S)< \left(\frac{4}{3}+\frac{8}{9r+12}\right) \cdot
c_w(\text{MRCT}(G))$$
for some $S.$
\end{fact}

The general star with a core in Fact~\ref{substruct}
can be computed unambiguously in logarithmic space, as the next lemma shows.
%The proof is included in the appendix.
\begin{lemma} \label{key}
%Let $M$ be a nondeterministic logarithmic-space Turing machine that, on
%input a weighted undirected graph $G=(V,E)$ with nonnegative edge-weight $w: E\to
%\mathbb{R}^+_0.$
Let $r\in\mathbb{N}$ be a constant.
There is a nondeterministic logarithmic-space Turing machine $\text{STAR}$ that, on input
a strongly min-unique connected graph $G=(V,E)$ with a
nonnegative edge-weight function $w: E\to \{0,\ldots,\text{poly}(|V|)\}$ and a sequence $S$
of at most $r+4$ vertices in $V,$
outputs $C_S=\text{CORE}(G,w,S)$ and each unique path in
$\text{SP}_G\left(u,\text{closest}(u,C_S)\right)$
for $u\in V\setminus V(C_S)$
unambiguously.
\end{lemma}

%We also need the following lemma.
The following lemma allows unambiguous logarithmic-space computation of the
routing cost of a tree.
%Its proof is included in the appendix.
\begin{lemma} \label{routpair}
%Let $G=(V,E)$ be a connected, min-unique, weighted undirected graph with nonnegative edge-weight $w: E\to
%\mathbb{R}^+_0$ and $s,t,x,y,z\in V.$
%Let $S=\text{CORE}(G,w,x,y,z).$
There is a nondeterministic logarithmic-space Turing machine ROUT-PAIR that, on input
a tree $T$ with a nonnegative
edge-weight function $w: E(T)\to \{0,\ldots,\text{poly}(|V(T)|)\}$ and $s,t\in V(T),$
unambiguously outputs the unique simple path $P^*$ connecting
$s$ and $t$ in $T$ and $w(P^*).$
\end{lemma}

Combining Fact~\ref{substruct} and Lemmas~\ref{key}--\ref{routpair} gives the following
lemma.
%Its proof is included in the appendix.
\begin{lemma} \label{unambiguousapprox}
Let $r\in\mathbb{N}$ be a constant.
There exists a nondeterministic logarithmic-space Turing machine APPROX that, on input
a strongly min-unique connected graph $G=(V,E)$ with a nonnegative edge-weight function $w: E\to
\{0,\ldots,\text{poly}(|V|)\},$
unambiguously outputs a spanning tree $T$ of $G$ with
$$c_w(T)< \left(\frac{4}{3}+\frac{8}{9r+12}\right) \cdot
c_w(\text{MRCT}(G)).$$
\end{lemma}
The following lemma will be useful.
%We include its proof in the appendix.
\begin{lemma}\label{smallerr}
Let $\alpha>0$ be a constant.
Let $G=(V,E)$ be a graph with a nonnegative edge-weight function $w:E\to
\mathbb{R}^+_0,$ and the minimum nonzero weight assigned by $w,$ if it exists,
is at least $1.$
Let $T_1$ and $T_2$ be spanning trees of $G.$
Let $w_r$ assign to each edge $e\in E$ a nonnegative weight $w(e)\le
1/{|V|^4}$ and $w^\prime=w+w_r.$
Then
\begin{eqnarray}
c_{w^\prime}(T_1)\le \alpha \, c_{w^\prime}(T_2) \label{eqA}
\end{eqnarray}
implies
\begin{eqnarray}
c_{w}(T_1)\le \alpha \, (1+\frac{1}{2|V|}) \, c_{w}(T_2) \label{eqB}
\end{eqnarray}
for sufficiently large $|V|.$
\end{lemma}
Combining Theorem~\ref{isolation} and
Lemma~\ref{unambiguousapprox}--\ref{smallerr} yields the following theorem.
%Its proof is included in the appendix.
\begin{theorem} \label{parallel}
Let $\epsilon>0$ be a constant.
There is an ${\cal RNC}^2$ algorithm PARALLEL that, on input a weighted undirected graph
$G=(V,E)$ with a nonnegative edge-weight function $w: E\to
\{0,\ldots,\text{poly}(|V|)\},$
satisfies the following.
\begin{itemize}
\item [1.] If $G$ is disconnected, then $\text{PARALLEL}(G,w)$ outputs
``disconnected.''
\item [2.] If $G$ is connected, then $\text{PARALLEL}(G,w)$
outputs a spanning tree of $G$ unambiguously or outputs ``fail''
unambiguously.
The probability that $\text{PARALLEL}(G)$ outputs a spanning tree of $G$
unambiguously is at least $1-1/{(2|V|)}.$
If $\text{PARALLEL}(G)$ outputs a spanning tree $T$ of $G$ unambiguously, then
$$c_{w}(T) \le \left(\frac{4}{3}+\epsilon\right) \cdot
c_{w}(\text{MRCT}(G)).$$
\end{itemize}
\end{theorem}

\comment{ % Begin removal due to space considerations
We can also turn the random weight assignment of PARALLEL into polynomially long
advices if we investigate the proof of Theorem~\ref{parallel}.
The method is similar to that in \cite{RA00}.
\begin{cor}
Let $\epsilon>0$ be a constant.
There is a ${\cal UL}/\text{poly}$ algorithm UNAMB-APPROX that,
on input a graph $G$ with a nonnegative edge-weight function
$w:E\to \mathbb{R}^+_0,$
satisfies the following.
\begin{itemize}
\item [1.] If $G$ is disconnected, then $\text{UNAMB-APPROX}(G,w)$ outputs
``disconnected.''
\item [2.] If $G$ is connected, then $\text{UNAMB-APPROX}(G,w)$
outputs a spanning tree $T$ of $G$ with
$$c_{w}(T) \le (\frac{4}{3}+\epsilon) \cdot c_{w}(\text{MRCT}(G)).$$
\end{itemize}
\end{cor}
} % End removal due to space considerations

%too easy

%\section{Variants} \label{variations}
\section{The SROCT problem} \label{SROCTsection}
We begin this section with the following definition.
%In this section we describe several variants of the MRCT problem and
%parallelize some proposed approximation algorithms for them.
\begin{definition} (\cite{WCT00PROCTSROCT, Wu02}) \label{variantdef}
Let $G=(V,E)$ be a graph with a nonnegative edge-weight function $w:E\to
\mathbb{R}_0^+$ and $r:V\to \mathbb{R}_0^+$ be a requirement function on
vertices.
Let $s_1,s_2\in V$ be two vertices of $G$ and $T$ be a spanning tree of $G.$
%Let $\lambda\ge 1.$
The sum-requirement communication (s.r.c.) cost of $T$ is
$$c^{(s)}_w(T)=\sum_{u,v\in V} (r(u)+r(v)) \, d_T(u,v).$$
The {\sc Sum-Requirement Optimal Communication Spanning Tree} (SROCT) problem
is to find a spanning tree $T$ of $G$ with the minimum value of $c^{(s)}_w(T)$
over all spanning trees of $G.$
We use $\text{SROCT}(G)$ to denote an arbitrary spanning
tree of $G$ with the minimum s.r.c. cost.
%The product-requirement communication (p.r.c.) cost of $T$ is
%$$c^{(p)}_w(T)=\sum_{u,v\in V} r(u) \, r(v) \, d_T(u,v).$$
\comment{%too easy
The product-requirement communication (p.r.c.) cost of $T$ is
$$c^{(p)}_w(T)=\sum_{u,v\in V} r(u) \, r(v) \, d_T(u,v).$$
The {\sc Product-Requirement Optimal Communication Spanning Tree} (PROCT) problem
is to find a spanning tree $T$ of $G$ with the minimum value of $c^{(p)}_w(T)$
over all spanning trees of $G.$
We use $\text{PROCT}(G)$ to denote an arbitrary spanning
tree of $G$ with the minimum p.r.c. cost.
}%too easy
The two-source routing cost of $T$ with sources $s_1,s_2$ is
$$c^{(2)}_w(T,s_1,s_2)=\sum_{v\in V} \left(d_T(s_1,v)+d_T(s_2,v)\right).$$
The $2$-MRCT problem is to find a spanning tree $T$ of $G$ with the minimum
value of $c^{(2)}_w(T,s_1,s_2)$ over all spanning trees of $G$ (in this
problem $s_1$ and $s_2$ are part of the input).
We use $\text{$2$-MRCT}(G)$ to denote an arbitrary spanning
tree of $G$ with the minimum two-source routing cost when the sources
$s_1,s_2$ are clear from the context.
Let $\lambda\ge 1.$
The weighted two-source routing cost of $T$ with sources $s_1,s_2$ and
weight $\lambda$ is
\begin{eqnarray*}
&&c^{(2)}_w(T,s_1,s_2,\lambda)\\
&=&\sum_{v\in V} \left(\lambda \, d_T(s_1,v)+d_T(s_2,v)\right).
\end{eqnarray*}
The weighted $2$-MRCT problem is to find a spanning tree $T$ of $G$ with the
minimum value of $c^{(2)}_w(T,s_1,s_2,\lambda)$ over all spanning trees of $G$ (in
this problem $s_1,s_2$ and $\lambda$ are part of the input).
We use $\text{W-$2$-MRCT}(G)$ to denote an arbitrary spanning
tree of $G$ with the minimum weighted two-source routing cost when $\lambda$
and the sources $s_1,s_2$ are clear from the context.
\end{definition}
The SROCT, $2$-MRCT and weighted $2$-MRCT problems are all
${\cal NP}$-hard, even on metric graphs \cite{WLBCRT00, WCT00PROCTSROCT,
WCT00PROCTPTAS, Wu02}.

%\subsection{The SROCT problem}
The following theorem gives a
$2$-approximation solution to the SROCT problem.
\begin{theorem} (\cite{WCT00PROCTSROCT}) \label{SROCTapprox}
Let $G=(V,E)$ be a connected graph with a nonnegative edge-weight function $w$ and
a nonnegative vertex-requirement function $r.$
There exists a vertex $x\in V$ such that any
shortest path tree $T$ rooted at $x$ satisfies
%general star $T$ with core the singleton $x,$
%we have
$$c^{(s)}_w(T) \le 2 c^{(s)}_w(\text{SROCT}(G)).$$
\end{theorem}
Theorems~\ref{isolation}--\ref{doublecount}, \ref{SROCTapprox} and
Lemma~\ref{routpair} give the following
parallelized $2$-approximation solution to the SROCT problem.
%Its proof is included in the appendix.
\begin{theorem} \label{parallelSROCT}
There is an ${\cal RNC}^2$ algorithm PARALLEL-SROCT that, on input a connected graph $G=(V,E)$ with
a nonnegative edge-weight function $w:E\to \{0,\ldots,\text{poly}(|V|)\}$ and a nonnegative vertex-requirement
function $r:V\to \{0,\ldots,\text{poly}(|V|)\},$ outputs a spanning $T$ of $G$ with
$$c^{(s)}_w(T)\le (2+o(1)) \, c^{(s)}_w(\text{SROCT}(G))$$
with high probability.
If PARALLEL-SROCT does not output such a spanning tree, it outputs ``fail.''
\end{theorem}

%too easy

\section{Weighted $2$-MRCT problem} \label{twoMRCTsection}
For the weighted $2$-MRCT problem, we can assume without loss of generality
that the two sources $s_1,s_2$ are such that $d_G(s_1,s_2)>0,$ where $G$ is
the input graph. Otherwise, the problem reduces to
finding a shortest path tree rooted at $s_1,$ which was implicitly done in the
proof of Theorem~\ref{parallelSROCT}.
Wu \cite{Wu02} has the following $2$-approximation solution for the weighted
$2$-MRCT problem.
\begin{theorem} (\cite{Wu02}) \label{twoMRCT}
Let $G=(V,E)$ be a connected graph with a nonnegative edge-weight function $w:E\to
\mathbb{R}^+_0,$ two sources $s_1,s_2\in V$ with $d_G(s_1,s_2)>0$ and $\lambda\ge 1.$
Denote $$D_1(v)=(\lambda+1)\, d_G(v,s_1) + d_G(s_1,s_2)$$
and
$$D_2(v)=(\lambda+1)\, d_G(v,s_2) + \lambda \, d_G(s_1,s_2)$$ for $v\in V.$
Let $Z_1^w=\{v\mid D_1(v)\le D_2(v)\}$ and $Z_2^w=V\setminus Z_1^w.$
Let $Q\in \text{SP}_G(s_1,s_2)$ be arbitrary.
Denote
\begin{eqnarray*}
Q=\left(q_0=s_1,\ldots,q_j,q_{j+1},\ldots,s_2\right)
\end{eqnarray*}
where $q_{j+1}$ is the first vertex on $Q$ (in the direction from $s_1$ to
$s_2$) that is not in $Z_1^w$ (it is easy to see that $s_1\in Z_1^w$).
For each $v\in V,$ let $P_{v,s_1} \in \text{SP}_G(v,s_1)$ and
$P_{v,s_2} \in \text{SP}_G(v,s_2)$ be arbitrary.
If $T_1=\bigcup_{v\in {Z_1^w}} P_{v,s_1}$ and $T_2=\bigcup_{v\in {Z_2^w}}
P_{v,s_2}$ are trees, then
$T=T_1\cup T_2\cup q_j q_{j+1}$ is a spanning tree of $G$
and
$$c^{(2)}_w(T)\le 2 c^{(2)}_w(\text{W-$2$-MRCT}(G)).$$
\end{theorem}
Theorems~\ref{isolation}--\ref{doublecount} and \ref{twoMRCT} and
Lemma~\ref{routpair}
yield the following theorem.
%whose proof is in the appendix.
\begin{theorem} \label{weightedtwoMRCT}
There is an ${\cal RNC}^2$ algorithm $\text{WEIGHTED-$2$-MRCT}$ that,
on input a graph $G=(V,E)$ with a nonnegative
edge-weight function $w:E\to \{0,\ldots,\text{poly}(|V|)\},$ $s_1,s_2\in V$ and $\lambda\ge 1,$ with high
probability outputs a spanning tree $T$ with
$$c^{(2)}_w(T)\le (2+o(1))\, c^{(2)}_w(\text{W-$2$-MRCT}(G)).$$
If WEIGHTED-$2$-MRCT does not output such a spanning tree, it outputs
``fail.''
\end{theorem}
We make the following concluding remark.
All our algorithms are shown to be ${\cal RNC}^2$-computable by showing
that they run in unambiguous logarithmic space and
succeed in giving an approximate solution when the random input specifies
an edge-weight function $w_r$ such that $G$ is strongly min-unique with
respect to $w+w_r.$
%We can also turn the random weight assignment of PARALLEL into polynomially long
%advices if we investigate the proof of Theorem~\ref{parallel}.
By a method similar to that in \cite{RA00}, we can also turn the random
weight assignment into polynomially long advices.
This is summarized below.
\begin{cor}
Let $\epsilon>0$ be a constant.
There are ${\cal UL}/\text{poly}$ algorithms for
$(4/3+\epsilon)$-approximating the MRCT problem, $(2+o(1))$-approximating
the SROCT problem and $(2+o(1))$-approximating the weighted $2$-MRCT
problem, where the respective edge-weight and vertex-requirement functions
are given in unary.
\end{cor}

\section{Conclusion} \label{conclusion}
We have given parallelized approximation algorithms for the
minimum routing
cost spanning tree problem and some of its variants.
Our results show that, by exhibiting multiple processors, we can compute
approximate solutions to the considered problems in
parallel poly-logarithmic time.
We hope this will shed light on the many areas in which
the considered problems are concerned, for example
network design \cite{Hu74, JLK78} and multiple sequences
alignment in computational biology \cite{FD87, Pev92, Gus93, BLP94, WLBCRT00}.

%Challenges remain to remove the need of randomness in our parallelized
%approximation algorithms. That is, to give \textbf{NC}-computable approximation algorithms for the
%MRCT problem. Another challenge is to turn our algorithms into
%${\cal RNC}^1$ circuits to further improve the level of parallelization.
%Still another challenge is to derive ${\cal RNC}$ approximation algorithms
%with approximation ratio $(1+\epsilon)$ on non-metric graphs.

\comment{ % begin removal
\subsection{Weighted $2$-MRCT on metric graphs}
The following theorem is due to Wu \cite{Wu02}.
\begin{theorem} (\cite{Wu02}) \label{metrictwoMRCT}
Let $\epsilon>0$ be a constant and $k=\lceil2/\epsilon-1\rceil.$
Let $G=(V,E)$ be a metric graph with a nonnegative edge-weight function $w$ and
$\lambda\ge 1.$
There is a sequence $(x_1,\ldots,x_k)\in V^k$ with the following property.
Consider the path $Q=(s_1,x_1,\ldots,x_k,s_2).$
For each $v\in V\setminus V(Q),$ let $i(v)$ be the smallest value of $i$ such that
$$(\lambda+1)\, w(vm_i) + \lambda \, d_Q(m_i,s_1) + d_Q(m_i,s_2)$$ is
minimized.
Then $T=Q\cup \bigcup_{v\in V\setminus V(Q)} m_{i(v)}v$ is a spanning tree of
$G$ satisfying
$$c^{(2)}_w(T)\le (1+\epsilon/2) \, c^{(2)}_w(\text{$2$-MRCT}(G)).$$
\end{theorem}
Theorems~\ref{isola}
} % end removal

\section*{Acknowledgments}
The authors are grateful to Wen-Hui Chen for his helpful comments and
suggestions.

\appendix
\section*{Appendix}
\begin{proof} [Proof of Theorem~\ref{isolation}.]
The theorem is clearly true for $|V|\le 2.$
If $G$ is not strongly min-unique with respect to $w+w_r$ and
$|V|\ge 3,$ we have seen that there exist $0\le k\le |V|-1, s,t\in V$
such that
$\left|\text{SP}^{(k)}_{s,t}\right|\ge 2$ where the edge weights are given with respect to
$w+w_r.$
This implies the existence of an edge $e\in E$ such that at least one
path in $\text{SP}^{(k)}_{s,t}$ contains $e,$ and at least one does not.
%shortest $s$-$t$ path
%(with respect to $w+w_r$) contains $e,$ and at least one does not.
In this case we say that $(k,s,t)$ blames $e.$
%Also, it is not hard to see that if some $(k,s,t)$ blames $e$ for $k\ge |V|,$
%then for some $k^\prime\in \{0,\ldots,|V|\}, s^\prime,t^\prime\in V$
%and $e^\prime\in E,$ we have that $(k^\prime,s^\prime,t^\prime)$ blames
%$e^\prime.$
Thus, the probability that $G$ is not strongly min-unique is at most the sum
over $0\le k\le |V|-1,$ $s,t\in V$ and $e\in E$ of the probability that
$(k,s,t)$ blames $e.$
\comment{ % Begin removal due to space requirement--replaced by text
\begin{eqnarray}
&&\Pr\left[G \mbox{ is not strongly min-unique}\right.\nonumber\\
&&\left. \mbox{ w.r.t. } w+w_r\right]\nonumber\\
&\le&\Pr\left[\exists k\in\{0,\ldots,|V|-1\}, s,t\in V \right.\nonumber\\
&&\left. \mbox{ s.t. }\left|\text{SP}^{(k)}_{s,t}\right|\ge 2
%\exists \mbox{ two shortest $s$-$t$ path w.r.t. }
\mbox{ w.r.t. } w+w_r\right] \nonumber\\
&\le& \sum_{0\le k\le |V|-1} \sum_{s,t\in V} \Pr\left[\exists e\in E, (k,s,t)
\right.\nonumber\\
&&\left. \mbox{ blames } e\right]\nonumber\\
&\le& \sum_{0\le k\le |V|-1; s,t\in V; e\in E} \Pr\left[(k,s,t)
\mbox{ blames } e\right].\label{eq1}
\end{eqnarray}
} % End removal due to space requirement--replaced by text

For any $k\in\{0,\ldots,|V|-1\}, s,t\in V,e\in E$
and any partial weight assignment of $w_r$ to $E\setminus\{e\},$
there is at most one assignment of $w_r$ to $e$ to make $(k,s,t)$ blame $e.$
This is because if $(k,s,t)$ blames $e$ when $w_r(e)$ is assigned a value
$w_{r,e},$ then increasing or decreasing the value of $w_r(e)$ forces all shortest
paths among those $s$-$t$ paths with at most $k$ edges
%(with respect to $w+w_r$)
to exclude or include $e,$
respectively, making $e$ no longer blamed by $(k,s,t).$
Therefore,
\begin{eqnarray*}
&&\sum_{0\le k\le |V|-1} \sum_{s,t\in V} \sum_{e\in E} \Pr\left[(k,s,t) \mbox{ blames } e\right]\\
&\le& \sum_{0\le k\le |V|-1} \sum_{s,t\in V} \sum_{e\in E} \frac{1}{|W|}\\
&\le& \frac{|V|^5}{2|W|},
\end{eqnarray*}
completing the proof.
%This and Eq.~(\ref{eq1}) complete the proof.
\end{proof}

%=====================================================================

\begin{proof} [Proof of Theorem~\ref{doublecount}.]
Before describing how FIND-PATH works, we describe a few procedures that are useful.
%Let $k\in \mathbb{N}.$
%Denote by $S_{k,s}\subseteq V$ the set of vertices reachable from $s$ using at most $k$ edges,
%and $c_k=|S_{k,s}|.$
%For $v\in S_{k,s},$
%let $\text{SP}^{(k)}_{s,v}$ be the set of all shortest $s$-$v$ paths among all $s$-$v$ paths with at most
%$k$ edges.
For $k\in\mathbb{N},$ let $c_k=|S_{k,s}|.$
Let $P^{(k)}_{s,v} \in \text{SP}^{(k)}_{s,v}$ be arbitrary and $\Sigma_k=\sum_{v\in
S_{k,s}} w(P^{(k)}_{s,v}).$
Note that the definition of $\Sigma_k$ does not depend on exactly which path in $\text{SP}^{(k)}_{s,v}$ is
chosen as $P^{(k)}_{s,v}.$
It is clear that $c_0=1$ and $\Sigma_0=0.$

%Given $c_k$ and $\Sigma_k$ and that $\left|\text{SP}^{(k)}_{s,v}\right|=1$
%for $v\in S_{k,s},$
%a key element of the proof is that we can compute $c_{k+1}$ and $\Sigma_{k+1}$ unambiguously,
%and determine whether $\left|\text{SP}^{(k+1)}(s,v)\right|>1$ for some $v\in V$
%unambiguously.
We first introduce a nondeterministic logarithmic-space subroutine OUTPUT that outputs $S_{k,s}$
unambiguously, given $G,w,s,c_k,\Sigma_k$ and that
$\left|\text{SP}^{(k)}_{s,v}\right|=1$ for each $v\in S_{k,s}.$
%For this purpose a subroutine that outputs $S_{k,s}$ unambiguously is needed.
%To output $S_{k,s}$ unambiguously, we
OUTPUT
just needs to nondeterministically guess each vertex $x$
to be in or out of $S_{k,s},$ and if the guess is $x\in S_{k,s}$ then it
outputs $x.$
It verifies each guess of $x\in S_{k,s}$ by
%To guarantee the guess to be right, we need to
nondeterministically guessing an $s$-$x$ path
with at most $k$ edges and rejecting if it fails.
%Clearly, guessing any vertex outside of $S_{k,s}$ to be in $S_{k,s}$ results
%in rejection.
Along the way OUTPUT counts the number $c_k^\prime$ of vertices verified to be in
$S_{k,s}$ and accumulates
the weights of the guessed $s$-$x$ paths (for $x$ verified to be in $S_{k,s}$)
in a variable $\Sigma_k^\prime.$
It then rejects if $c_k^\prime\neq c_k$ or $\Sigma_k^\prime\neq \Sigma_k.$
Clearly,
guessing any vertex out of $S_{k,s}$ to be in $S_{k,s}$ results in rejection.
For a computation branch of OUTPUT not to reject, it must
have $c_k^\prime$ reach $c_k,$ which requires successfully guessing
an $s$-$x$ path with at most $k$ edges for each $x\in S_{k,s}.$
But to have $\Sigma_k^\prime$ not exceed $\Sigma_k,$
the guessed $s$-$x$ path for each $x\in S_{k,s}$ should be
the unique one in $\text{SP}^{(k)}_{s,x}.$
%shortest, and such paths are unique by the assumption that
%$\left|\text{SP}^{(k)}_{s,v}\right|=1$ for $v\in S_{k,s}.$
So $\text{OUTPUT}(G,w,s,c_k,\Sigma_k)$ has a
%outputs $S_{k,s}$ unambiguously.
%Clearly, along the
unique non-rejecting computation branch, on which it
%outputs $S_{k,s},$ OUTPUT
correctly guesses whether each vertex $x$
belongs to $S_{k,s}$ and if so, correctly guesses the unique path in $\text{SP}^{(k)}_{s,x}.$

We now describe a procedure INDUCTIVE that computes $c_{k+1}$ and
$\Sigma_{k+1},$ and determines whether
$\left|\text{SP}^{(k+1)}_{s,x}\right|>1$ for some $x\in V$ unambiguously,
given $c_k$ and $\Sigma_k$ and that $\left|\text{SP}^{(k)}_{s,v}\right|=1$ for
each $v\in S_{k,s}.$
For each vertex $x\in V,$ INDUCTIVE runs
$\text{OUTPUT}(G,w,s,c_k,\Sigma_k)$ to
determine whether $x\in S_{k,s}$ unambiguously and if so, accumulates the
weight $w_x$ of the unique path in $\text{SP}^{(k)}_{s,x}$ as it is guessed
by OUTPUT.
For each $u$ such that $ux$ is an edge, INDUCTIVE also determines whether $u\in
S_{k,s}$ unambiguously and if so, accumulates the
weight $w_u$ of the unique path in $\text{SP}^{(k)}_{s,u}$ as it is guessed.
INDUCTIVE then computes the weight of any shortest $s$-$x$ path with at most $k+1$ edges as
$$\min Q$$
where $$Q=\{w_x\}\cup \{w_u+w(ux)\mid u\in S_{k,s}, ux\in E\}$$ if $x\in S_{k,s}$
and $$Q=\{w_u+w(ux)\mid u\in S_{k,s}, ux\in E\}$$ otherwise.
If $Q=\emptyset,$ INDUCTIVE knows that $x\notin S_{k+1,s}.$
Otherwise $x\in S_{k+1,s}$ and the weight of any path in
$\text{SP}^{(k+1)}_{s,x}$ is known to be $\min Q.$
In case of a tie when computing $\min Q,$
%$$\min\{w_x\}\cup\{w_u+w(ux)\mid u\in S_{k,s}\},$$
INDUCTIVE knows that $\left|\text{SP}^{(k+1)}_{s,x}\right|>1.$
Doing the above for all $x\in V$ allows INDUCTIVE to compute
$c_{k+1}=|S_{k+1,s}|$ and $\Sigma_{k+1}$
and determine whether
$\left|\text{SP}^{(k+1)}_{s,x}\right|>1$ for some $x\in V$ unambiguously.
During the computation of $\min Q,$
INDUCTIVE does not store the set $Q.$
Instead, INDUCTIVE computes the elements of $Q$ one by one and
stores the smallest element in $Q$ that has been computed so far,
as well as a flag indicating whether $\min Q$ is achieved by
two elements at any time.

We are now ready to describe how FIND-PATH works. Assume $|V|\ge 3.$
FIND-PATH starts with $c_0, \Sigma_0$ and repeatedly simulates INDUCTIVE
until it computes $c_{|V|},\Sigma_{|V|}$ or until it determines that
$\left|\text{SP}^{(k)}_{s,x}\right|>1$ for some $0\le k\le |V|-1$
and $x\in V.$
Doing the above with each other vertex $s^\prime\in V$ replacing the role of $s$
guarantees that if $G$ is not
strongly min-unique, then
we must find $\left|\text{SP}^{(k)}_{s^\prime,x}\right|>1$ for some $0\le
k\le |V|-1$ and $s^\prime,x\in V.$
%and thus find out that $G$ as not strongly min-unique.
Instead, if $G$ is strongly min-unique
then FIND-PATH will compute all the way from $c_0,\Sigma_0$ to
$c_{|V|},\Sigma_{|V|}.$
It then runs $\text{OUTPUT}(G,w,s,c_{|V|-1},\Sigma_{|V|-1}).$
As we have seen, $\text{OUTPUT}(G,w,s,c_{|V|-1},\Sigma_{|V|-1})$ has a unique
non-rejecting computation branch, on which the unique shortest $s$-$t$ path $P\in
\text{SP}^{(|V|-1)}_{s,t}$ is correctly guessed by OUTPUT.
The weight $w(P)$ is accumulated along the way.
The strong min-uniqueness of $G$ guarantees that $P$ is also the unique path in $\text{SP}_G(s,t).$
%we have shown that we can output
%$S_{k,s}$ unambiguously and on the unique non-rejecting computation branch, we successfully guess a shortest
%$s$-$x$ path for each vertex $x\in S_{k,s}.$
%Monitoring the output for a shortest $s$-$t$ path and discarding other outputs does the job.
\end{proof}

%=================================================================

\begin{proof} [Proof of Lemma~\ref{add}.]
For each $x\in V(R),$ ADD-PATH runs $\text{FIND-PATH}(G,w,x,v)$ to unambiguously
generate the unique shortest path $P_{x,v}\in\text{SP}_G(x,v)$ and
its weight $w(P_{x,v}).$
In this way, ADD-PATH could compute $\min_{x\in V(R)} w(P_{x,v})$ as well as
$\text{closest}(v,R).$
Then ADD-PATH simulates
$\text{FIND-PATH}\left(G,w,v,\text{closest}(v,R)\right)$ to output the unique
path $P$ in $\text{SP}_G\left(v,\text{closest}(v,R)\right)$ unambiguously.
%If $P$ visits $R$ twice, then there are vertices $y,z\in V(R)$ such that
%there is a $y$-$z$ path of zero weight, and thus $G$ is not strongly min-unique. A
%contradiction.
\end{proof}

%=================================================================

\begin{proof} [Proof of Lemma~\ref{core}.]
%Note that for $1\le j\le k,$
%$$R_{j,S}=R_{1,S} \bigcup_{2\le i\le j} P_{i,S}.$$
Clearly, CORE could output $R_{1,S}$ unambiguously.
Let $2\le j\le k.$
To output $R_{j,S}=R_{j-1,S}\cup P_{j,S}$ unambiguously, CORE recursively outputs $R_{j-1,S}$ and
then runs
$\text{ADD-PATH}(G,w,R_{j-1,S},v_j)$ to unambiguously output $P_{j,S}.$
There is an additional complication that CORE does not store $R_{j-1,S}$
before calling ADD-PATH.
Instead, whenever ADD-PATH wants to read any bit encoding $R_{j-1,S},$
CORE recursively outputs $R_{j-1,S}$ unambiguously on the fly to support the
required bit.
%and then clears the space used for recursively outputting $R_{j-1,S}.$
Each level of recursion uses up logarithmic space and the depth of recursion
is at most $r+4,$ a constant.
The space requirement is therefore logarithmic.

%If $P_{i,S}$ visits $R_{i-1,S}$ twice, then there are vertices $y,z\in
%V(R_{i-1,S})$ with a $y$-$z$ path of zero weight, and thus $G$ is not
%strongly min-unique. A contradiction.
%Assume that CORE could compute $R_{i-1,S}$ unambiguously, then it needs only
%run $\text{ADD-PATH}(G,w,R_{i-1,S},v_i)$ to output 
\end{proof}

%=================================================================

\begin{proof} [Proof of Lemma~\ref{key}.]
%Testing if $G$ is connected is done by testing whether every two
%vertices are connected, which can be done in logarithmic space \cite{Rei05}.
%Thus item $1$ is satisfied.
%By Theorem~\ref{doublecount}, if $G$ is not min-unique then simulating
%$\text{FIND-PATH}(G,w,x,y)$ unambiguously determines this fact.
%Hence item $2$ is also easy.
%
%We now assume that $G$ is connected and min-unique and establish item $3.$
STAR begins by running $\text{CORE}(G,w,S)$ to output $C_S$ unambiguously.
For any $u\in V,$ STAR runs $\text{CORE}(G,w,S)$ to unambiguously
determine whether $u\in V(C_S).$
If $u\notin V(C_S),$ STAR needs to output the unique path in
$\text{SP}_G\left(u,\text{closest}(u,C_S)\right).$
For this purpose, it computes $\text{closest}(u,C_S)$ as follows.
For each $v\in V,$
STAR tests if $v\in V(C_S),$ again by running $\text{CORE}(G,w,S).$
If $v\notin V(C_S),$ STAR goes on with the next $v\in V.$
Otherwise, STAR invokes $\text{FIND-PATH}(G,w,u,v)$ to generate the unique
path $P_{u,v} \in \text{SP}_G(u,v)$ and its weight $w(P_{u,v})$
unambiguously.
%while accumulating its weight $w(P_{u,v})$ along the way.
STAR
records the $v\in V(C_S)$ that has generated the smallest value of $w(P_{u,v})$ so
far, favoring lexicographically smaller values of $v$ in case of a tie.
%STAR then goes on with the next $v\in V.$
In the end, the recorded $v\in V$ must be $\text{closest}(u,C_S)$
by the definition of $\text{closest}(u,C_S).$
%After trying all $y\in V,$ the vertex $\text{closest}(x,P)$ must have been
%obtained.
At this time STAR just invokes
$\text{FIND-PATH}\left(G,w,u,\text{closest}(u,C_S)\right)$ to
output the unique path in $\text{SP}_G\left(u,\text{closest}(u,C_S)\right)$ unambiguously.
Doing the above for all $u\in V$ does the job.
%By recycling space,
%the space complexity is $O(\log{|V|})$ plus that of running FIND-PATH and
%CORE, which is $O(\log |V|).$
\end{proof}

%=================================================================

\begin{proof} [Proof of Lemma~\ref{routpair}.]
If $s=t$ the task is trivial. We assume otherwise.
ROUT-PAIR nondeterministically guesses a path $P$ that does not enter a
vertex immediately after it has left that vertex.
If $P$ is an $s$-$t$ path, then ROUT-PAIR accepts, otherwise it rejects.
The simple $s$-$t$ path in $T$ is
the only $s$-$t$ path that does not enter a vertex immediately after leaving
it.
Its weight $w(P^*)$ can be accumulated as it is guessed.
%The path $P$ is guessed on the fly so that the space requirement is
%logarithmic in $|V|.$
%When an edge $e$ of $P$ is guessed, ROUT-PAIR simulates
%$\text{STAR}(G,w,x,y,z)$ to determine if $e$ is on $T_{x,y,z}.$
%If so, ROUT-PAIR outputs $e$ and continues execution; otherwise it rejects.
%To argue that ROUT-PAIR unambiguously outputs $P^*,$ note that
%$P^*$ is the only $s$-$t$ path in $T_{x,y,z}$
%that does not enter a node immediately after leaving it.
\end{proof}

%=================================================================

\begin{proof} [Proof of Lemma~\ref{unambiguousapprox}.]
%Items $1$ and $2$ are as in Lemma~\ref{key}.
%We now assume that $G$ is connected and min-unique and deal with item $3.$
For each sequence $S$ of at most $r+4$ vertices in $V,$
Lemma~\ref{key} enables us to unambiguously output
$C_S=\text{CORE}(G,w,S)$ and then each unique path $P_u$ in
$\text{SP}_G\left(u,\text{closest}(u,C_S)\right)$ for $u\in V\setminus V(C_S).$
Furthermore, Fact~\ref{substruct} guarantees that
$$T_S=C_S \cup \bigcup_{u\in V\setminus V(C_S)} P_u$$
%\text{SP}_G(u,\text{closest}(u,G))$$
is a spanning tree of $G$ and
satisfies $$c_w(T_S)< \left(\frac{4}{3}+\frac{8}{9r+12}\right) \cdot
c_w(\text{MRCT}(G))$$
for some $S.$
Thus, we need only compute $c_w(T_S)$ unambiguously for each sequence $S$ of
at most $r+4$ vertices,
and output $T_{S^*}$ unambiguously for the sequence $S^*$ of at most $r+4$
vertices satisfying
$c_w(T_{S^*})=\min_{S\in V^k, 1\le k\le r+4} c_w(T_S).$
By Definition~\ref{routingcost},
$c_w(T_S)$ can be computed unambiguously by running
$\text{ROUT-PAIR}(T_S,w,s,t)$ for all pairs
$s,t \in V$ and summing up the weight of the simple paths as they are
output.
There is a complication that APPROX does not store $T_S$ before calling
ROUT-PAIR.
Instead, when ROUT-PAIR wants to read any bit in the encoding of $T_S,$
APPROX runs $\text{STAR}(G,w,S)$ to generate the required bit unambiguously
on the fly.
%The unique computation branch that has not rejected will have the value
%$c(T_{u,v})$ by the Definition~\ref{routingcost}.
This enables us to obtain $S^*$ unambiguously and thus
$T_{S^*}$ unambiguously by running $\text{STAR}(G,w,S^*).$
%with one additional complication that a vertex or edge
%may be output by $\text{STAR}(G,w,x^*,y^*,z^*)$
%more than once on its unique non-rejecting computation branch.\footnote{Because it outputs
%things like $\text{closest}(u,S)$ over all $u\in V\setminus V(S).$ These may
%contain repeats.}
%Eliminating repeated outputs of vertices or edges is easy if, before
%outputting any vertex or edge, we generate $T_{x^*,y^*,z^*}$ unambiguously
%again to check for repeated outputs of that vertex or edge.
\end{proof}

%=================================================================

\begin{proof} [Proof of Lemma~\ref{smallerr}.]
%Assume $c_{w^\prime}(T_1)\le \alpha \, c_{w^\prime}(T_2).$
It is clear that either $c_{w}(T_2)=0$ or $c_{w}(T_2)\ge 1.$
Also, Fact~\ref{loading} implies that
\begin{eqnarray}
c_{w_r}(T_2)\le \frac{1}{2|V|}.\label{smallpart}
\end{eqnarray}

If $c_{w}(T_2)=0,$
Eq.~(\ref{smallpart}) implies that
$$c_{w^\prime}(T_2)=c_{w_r}(T_2)\le \frac{1}{2|V|}$$
and thus $c_{w^\prime}(T_1)<1$ for sufficiently large $|V|$ by
Eq.~(\ref{eqA}).
This implies $c_{w}(T_1)<1$ and thus $c_{w}(T_1)=0,$ establishing
Eq.~(\ref{eqB}).

If $c_{w}(T_2)\ge 1,$ then Eq.~(\ref{eqA})
%$$c_{w^\prime}(T_1)\le \alpha \, c_{w}(T_2)+c_{w_r}(T_2)$$
and (\ref{smallpart}) imply
\begin{eqnarray*}
&&c_{w}(T_1)\le c_{w^\prime}(T_1) \\
&\le& \alpha \left(c_w(T_2)+c_{w_r}(T_2)\right)\\
&\le& \alpha \, (1+\frac{1}{2|V|}) \, c_{w}(T_2).
\end{eqnarray*}
\end{proof}

%=================================================================

\begin{proof} [Proof of Theorem~\ref{parallel}.]
We will show that PARALLEL needs only take a $\text{poly}(|V|)$-long random
input and do the rest of the computation in unambiguous logarithmic
space.
The standard proof technique for showing that ${\cal UL} \subseteq {\cal NL} \subseteq
{\cal NC}^2$ \cite{Pap94CC, Sip05} then completes the proof.

PARALLEL tests the connectedness of $G$ by testing each pair of vertices for
connectedness in logarithmic space \cite{Rei05}.
%PARALLEL first tests whether $G$ is disconnected and outputs ``disconnected''
%if so.
%This can be done in logarithmic space \cite{Rei05}.

Below we assume that $G$ is connected.
If we assume that this theorem is true when $w$ is not identically zero,
then PARALLEL can also deal with the identically zero case
by using the unit edge-weight function instead.
The output spanning tree would have zero routing cost under the identically
zero edge-weight function, so item $2$ is still satisfied.
Thus, we can assume without loss of generality that $w$ is not identically
zero.
Furthermore, we can normalize $w$ to give $\min_{e\in E,w(e)\neq 0} w(e)\ge 1.$

The random input to
PARALLEL determines an edge-weight function $w_r: E\to \mathbb{R}^+_0$
where for each $e\in E,$ $w_r(e)$ is independently and randomly chosen from
the uniform
%by independently and randomly choosing $w_r(e)$ from the uniform
distribution over $\{1/{|V|}^{10},\ldots,{|V|}^6/{|V|}^{10}\}.$
%for each edge $e\in E.$
Note that $w_r(e)\le 1/{|V|^4}$ for every $e\in E.$
%randomly and independently chooses the edge weights using the uniform
%distribution over $\{1,\ldots,{|V|}^5\}$ for each edge.
%Call this edge-weight function $w.$
%We use $c_1(\cdot)$ for the routing cost function with respect to the unit
%edge-weight and $c_w(\cdot)$ for that with respect to $w.$
Denote $w^\prime=w+w_r.$
Let $\hat{T}_w$ be an MRCT with respect to $w$ and $\hat{T}_{w^\prime}$ be
that with respect to $w^\prime.$
By Theorem~\ref{isolation}, with probability at least $1-1/{(2|V|)},$ $G$ is
strongly min-unique with respect to $w^\prime.$
PARALLEL runs $\text{FIND-PATH}(G,w^\prime,s,t)$ for an arbitrary pair $s,t\in V$
to unambiguously test if $G$ is strongly min-unique with respect to $w^\prime,$ and outputs ``fail'' if it
is not.
PARALLEL then
runs $\text{APPROX}(G,w^\prime)$ to unambiguously output a tree $T^\prime$ with
\begin{eqnarray}
& &c_{w^\prime}(T^\prime) < \left(\frac{4}{3}+\epsilon/2\right) \cdot
c_{w^\prime}(\hat{T}_{w^\prime})\nonumber \\
&\le& \left(\frac{4}{3}+\epsilon/2\right)
\cdot c_{w^\prime}(\hat{T}_w) \label{approxEq}
\end{eqnarray}
by invoking Lemma~\ref{unambiguousapprox} with a sufficiently large constant $r$ such
that $8/(9r+12)< \epsilon/2.$

We shall prove that
\begin{eqnarray}
c_{w}(T^\prime) \le (\frac{4}{3}+\epsilon) \cdot c_{w}(\hat{T}_w),\label{goal}
\end{eqnarray}
which is true by Lemma~\ref{smallerr} for sufficiently large $|V|.$
% end removal
\end{proof}

%too easy

%=============================================================

\begin{proof} [Sketch of proof of Theorem~\ref{parallelSROCT}.]
We omit the simple case where $w$ is identically zero and assume without
loss of generality that
$\min_{e\in E, w(e)\neq 0} w(e) \ge 1.$
Let $G_0$ be the subgraph of $G$ formed by the zero-weight edges of $G.$
PARALLEL-SROCT tests whether $G_0$ is a connected spanning subgraph of $G$
in logarithmic space
\cite{Rei05} and if so, outputs a spanning tree of $G_0$ by calling, say,
$\text{PARALLEL}(G_0,0)$ where $0$ denotes the identically zero function.
%We assume without loss of generality that $\min_{e\in E} w(e)
%\ge 1.$

Below we assume that $G_0$ is disconnected.
%Let $w_r$ independently and randomly assign to each edge of $G$ a weight in
%$\{1/{|V|^{10}},\ldots,|V|^6/{|V|^{10}}\}$ and $w^\prime=w+w_r.$
The random input to PARALLEL-SROCT determines an edge-weight function
$w_r: E\to \mathbb{R}_0^+$ where for each $e\in E,$ $w_r(e)$ is
independently and randomly chosen from the uniform distribution over
$\{1/{|V|^{10}},\ldots,|V|^6/{|V|^{10}}\}.$
Let $w^\prime=w+w_r.$
Note that $\max_{e\in E} w_r(e)\le 1/|V|^4.$
By Theorem~\ref{isolation}, $G$ is strongly min-unique with respect to $w^\prime$
with high probability.
PARALLEL-SROCT uses FIND-PATH to determine if $G$ is strongly min-unique
with respect to $w^\prime$ and outputs ``fail'' if it is not.
Below we assume that $G$ is strongly min-unique with respect to $w^\prime.$
%$G$ is strongly min-unique with respect to $w^\prime.$
For each $x\in V,$ a shortest path tree $T_x$ rooted at $x$ can be output
unambiguously by running $\text{FIND-PATH}(G,w^\prime,x,y)$ for each $y\in V.$
The s.r.c. cost of $T_x$ can be computed unambiguously
by running $\text{ROUT-PAIR}(T_x,w^\prime,s,t)$ and
accumulating the weight of the output path multiplied by $r(s)+ r(t)$
for all $s,t\in V.$
During the computation of ROUT-PAIR,
each time any bit encoding $T_x$ is needed, it is generated on the fly.
The spanning tree $T_x$ with the minimum (over $x\in V$) s.r.c. cost with respect
to $w^\prime$ is output.
The final step in establishing the approximation ratio goes by showing that
for every spanning tree $T$ of $G,$
$$c^{(s)}_w(T)\ge \max_{v\in V} r(v)$$
by the disconnectedness of $G_0$ and $\min_{e\in E,w(e)\neq 0} w(e)\ge 1,$ whereas
$$c^{(s)}_{w_r}(T) \le \max_{u,v\in V} (r(u)+r(v)) \, \frac{|V|^3}{2}
\frac{1}{|V|^4}$$ by Fact~\ref{loading}.
The fact that ${\cal UL} \subseteq {\cal NL} \subseteq {\cal NC}^2$
completes the proof.
%The standard proof technique for showing that ${\cal UL} \subseteq {\cal NL} \subseteq
%{\cal NC}^2$ \cite{Pap94CC, Sip05} completes the proof.
\end{proof}

%==================================================================

\comment{%too easy
\begin{proof} [Proof of Theorem~\ref{parallelPROCT}.]
%The ${\cal RNC}^2$ circuit METRIC-PROCT tries each pair $x,y \in V$ in
%parallel.
For each pair $x,y\in V,$ a minimum $s$-$t$ cut $(V_1,V_2)$ of
$H_{x,y}$ can be
constructed by Theorem~\ref{mincutRNC}.
The p.r.c. cost of the tree with vertex set $V$
and edge set
\begin{eqnarray*}
&&\{xv\mid v\in V_1, v\neq x\}\\
&\cup& \{yv \mid v\in V_2,v\neq y\}
\cup \{xy\}
\end{eqnarray*}
can be computed in logarithmic space and the one with the minimum cost can
be selected for output.
The approximation ratio is guaranteed by Theorem~\ref{metricPROCT}.
%The standard proof technique for showing that
%${\cal UL} \subseteq {\cal NL} \subseteq {\cal NC}^2$ \cite{Pap94CC, Sip05} completes the proof.
The fact that ${\cal UL} \subseteq {\cal NL} \subseteq {\cal NC}^2$
completes the proof.
\end{proof}
}%too easy

%===============================================================

\begin{proof} [Sketch of proof of Theorem~\ref{weightedtwoMRCT}.]
We omit the simple case when $w$ is identically zero and assume that
$\min_{e\in E, w(e)\neq 0} w(e)\ge 1.$
The case where the zero-weight edges of $G$ form a connected spanning
subgraph of $G$ is dealt with as in the proof of
Theorem~\ref{parallelSROCT}, so we may assume that it is not the case.
%We assume without loss of generality that $\min_{e\in E} w(e)\ge 1.$
%Let $w_r$ independently and randomly assign to each edge of $G$ a weight in
%$\{1/{|V|^{10}},\ldots,|V|^6/{|V|^{10}}\}$ and $w^\prime=w+w_r.$
The random input to WEIGHTED-$2$-MRCT determines an edge-weight function
$w_r: E\to \mathbb{R}_0^+$ where for each $e\in E,$ $w_r(e)$ is
independently and randomly chosen from the uniform distribution over
$\{1/{|V|^{10}},\ldots,|V|^6/{|V|^{10}}\}.$
Let $w^\prime=w+w_r.$ Note that $\max_{e\in E} w_r(e)\le 1/{|V|}^4.$
WEIGHTED-$2$-MRCT detects whether $G$ is strongly min-unique with respect to
$w^\prime$ by running FIND-PATH and outputs ``fail'' if it is not, which
occurs with a small probability by Theorem~\ref{isolation}.
%By Theorem~\ref{isolation}, $G$ is strongly min-unique with respect to $w^\prime$
%with high probability.
Now, assume that $G$ is strongly min-unique with respect to $w^\prime.$
The sets $Z^{w^\prime}_1,Z^{w^\prime}_2$
in Theorem~\ref{twoMRCT} where $d_G(\cdot)$ is measured with
respect to $w^\prime$ are computable in unambiguous logarithmic space by Theorem~\ref{doublecount}.
For each $v\in V,$ let $P^{(w^\prime)}_{v,s_1}\in \text{SP}_G(v,s_1)$ and
$P^{(w^\prime)}_{v,s_2}\in \text{SP}_G(v,s_2)$ be the unique shortest paths
with respect to $w^\prime.$
By Theorem~\ref{doublecount},
$T_1=\bigcup_{v\in Z^{w^\prime}_1} P^{w^\prime}_{v,s_1}$ and $T_2=\bigcup_{v\in
Z^{w^\prime}_2} P^{w^\prime}_{v,s_2}$
are unambiguously computable in logarithmic space and they are
trees by the strong min-uniqueness of $G$ with respect to $w^\prime.$
The unique shortest path
$Q^{w^\prime}=\left(q_0=s_1,\ldots,q_j,q_{j+1},\ldots,s_2\right)\in \text{SP}_G(s_1,s_2)$ with respect
to $w^\prime$ is also unambiguously computable by running
$\text{FIND-PATH}(G,w^\prime,s_1,s_2),$ so is its first vertex $q_{j+1}$ outside of
$Z^{w^\prime}_1.$
%Therefore,
Theorem~\ref{twoMRCT} then implies that
a tree $T$ satisfying
$$c^{(2)}_{w^\prime}(T) \le 2
c^{(2)}_{w^\prime}(\text{W-$2$-MRCT}_{w^\prime}(G))$$
can be output in unambiguous logarithmic space.
%where $\text{W-$2$-MRCT}_{w^\prime}(G)$ is the spanning tree with the minimum
%weighted two-source routing cost with respect to $w^\prime.$
The final step in establishing the approximation ratio is to show that for every spanning tree $T$ of
$G,$
\begin{eqnarray}
c^{(2)}_w(T,s_1,s_2,\lambda)\ge \lambda+1 \label{oneedge}
\end{eqnarray}
whereas
$$c^{(2)}_{w_r}(T,s_1,s_2,\lambda) \le \lambda \,
\frac{|V|^2}{|V|^4}+\frac{|V|^2}{|V|^4}.$$
%The reason for Eq.~(\ref{oneedge}) to hold is that there must be some $v\in
%V$ with $d_T(s_1,v)\neq 0$ with respect to $w$ for otherwise $w$ would be
%identically zero. But $d_T(s_1,v)\neq 0$ (w.r.t. $w$) means $d_T(s_1,v)\ge 1$
%(w.r.t. $w$) by the assumption that $\min_{e\in E,w(e)\neq 0} w(e)\ge 1.$
%The standard proof technique for showing that ${\cal UL} \subseteq {\cal NL} \subseteq
%{\cal NC}^2$ \cite{Pap94CC, Sip05} completes the proof.
The fact that ${\cal UL} \subseteq {\cal NL} \subseteq {\cal NC}^2$
completes the proof.
\end{proof}

\bibliographystyle{amsalpha}
\bibliography{par}

\newcommand{\etalchar}[1]{$^{#1}$}
\providecommand{\bysame}{\leavevmode\hbox to3em{\hrulefill}\thinspace}
\providecommand{\MR}{\relax\ifhmode\unskip\space\fi MR }
\providecommand{\MRhref}[2]{%
  \href{http://www.ams.org/mathscinet-getitem?mr=#1}{#2}
}
\providecommand{\href}[2]{#2}
\begin{thebibliography}{WLB{\etalchar{+}}00}

\bibitem[Bar98]{Bar98}
Y.~Bartal, \emph{On approximating arbitrary metrics by tree metrics},
  Proceedings of the 30th Annual ACM Symposium on Theory of Computing, 1998,
  pp.~161--169.

\bibitem[BFW73]{BFW73}
D.~E. Boyce, A.~Farhi, and R.~Weischedel, \emph{Optimal network problem: A
  branch-and-bound algorithm}, Environment and Planning \textbf{5} (1973),
  519--533.

\bibitem[BLP94]{BLP94}
V.~Bafna, E.~L. Lawler, and P.~Pevzner, \emph{Approximations algorithms for
  multiple sequence alignment}, Proceedings of the 5th Combinatorial Pattern
  Matching Conference, vol. 807, 1994, pp.~43--53.

\bibitem[CCGG98]{CCGG98}
M.~Charikar, C.~Chekuri, A.~Goel, and S.~Guha, \emph{Rounding via trees:
  Deterministic approximation algorithms for group \uppercase{S}teiner trees
  and k-median}, Proceedings of the 30th Annual ACM Symposium on Theory of
  Computing, 1998, pp.~114--123.

\bibitem[CLRS01]{CLRS01}
T.~H. Cormen, C.~E. Leiserson, R.~L. Rivest, and C.~Stein, \emph{Introduction
  to algorithms.}, MIT Press, 2001.

\bibitem[DF79]{DF79}
R.~Dionne and M.~Florian, \emph{Exact and approximate algorithm for optimal
  network design}, Networks \textbf{9} (1979), no.~1, 37--60.

\bibitem[FD87]{FD87}
D.~Feng and R.~Doolittle, \emph{Progressive sequence alignment as a
  prerequisite to correct phylogenetic trees}, Journal of Molecular Evolution
  \textbf{25} (1987), 351--360.

\bibitem[Fis01]{Fis01}
Fischer, \emph{The art of uninformed decisions: {A} primer to property
  testing}, Bulletin of the European Association for Theoretical Computer
  Science \textbf{75} (2001).

\bibitem[FLS02]{FLS02}
M.~Fischetti, G.~Lancia, and P.~Serafini, \emph{Exact algorithms for minimum
  routing cost trees}, Networks \textbf{39} (2002), 161--173.

\bibitem[FRT03]{FRT03}
J.~Fakcharoenphol, S.~Rao, and K.~Talwar, \emph{A tight bound on approximating
  arbitrary metrics by tree metrics}, Proceedings of the 35th Annual ACM
  Symposium on Theory of Computing, 2003, pp.~448--455.

\bibitem[GJ79]{GJ79}
M.~R. Garey and D.~S. Johnson, \emph{Computers and intractability: A guide to
  the theory of \uppercase{${\cal NP}$}-completeness}, Freeman, 1979.

\bibitem[Gus93]{Gus93}
D.~Gusfield, \emph{Efficient methods for multiple sequence alignment with
  guaranteed error bounds}, Bulletin of Mathematical Biology \textbf{55}
  (1993).

\bibitem[GW96]{GW96}
A.~G\'al and A.~Wigderson, \emph{Boolean vs. arithmetic complexity classes:
  randomized reductions}, Random Structures and Algorithms \textbf{9} (1996),
  99--111.

\bibitem[Hoa73]{Hoa73}
H.~H. Hoang, \emph{A computational approach to the selection of an optimal
  network}, Management Science \textbf{19} (1973), no.~5, 488--498.

\bibitem[Hu74]{Hu74}
T.~C. Hu, \emph{Optimum communication spanning trees}, SIAM Journal on
  Computing \textbf{3} (1974), 188--195.

\bibitem[Imm88]{Imm88}
N.~Immerman, \emph{Nondeterministic space is closed under complementation},
  SIAM Journal on Computing (1988), 935--938.

\bibitem[JLK78]{JLK78}
D.~S. Johnson, J.~K. Lenstra, and A.~H. G.~Rinnooy Kan, \emph{The complexity of
  the network design problem}, Networks \textbf{8} (1978), 279--285.

\bibitem[Pap94]{Pap94CC}
C.~H. Papadimitriou, \emph{Computational complexity}, Addison-Wesley, 1994.

\bibitem[Pev92]{Pev92}
P.~Pevzner, \emph{Multiple alignment, communication cost, and graph matching},
  SIAM Journal on Applied Mathematics \textbf{52} (1992), 1763--1779.

\bibitem[RA00]{RA00}
K.~Reinhardt and E.~Allender, \emph{Making nondeterminism unambiguous}, SIAM
  Journal on Computing \textbf{29} (2000), no.~4, 1118--1131.

\bibitem[Rei05]{Rei05}
O.~Reingold, \emph{Undirected \uppercase{ST}-connectivity in log-space},
  Proceedings of the 37th Annual ACM Symposium on Theory of Computing, 2005,
  pp.~376--385.

\bibitem[Sip05]{Sip05}
M.~Sipser, \emph{Introduction to the theory of computation}, 2nd ed., Course
  Technology, 2005.

\bibitem[Sze88]{Sze88}
R.~Szelepcs{\'e}nyi, \emph{The method of forced enumeration for
  nondeterministic automata}, Acta Informatica \textbf{26} (1988), 279--284.

\bibitem[WC04]{CW04}
B.~Y. Wu and K.-M. Chao, \emph{Spanning trees and optimization problems},
  Chapman \& Hall/CRC Press, 2004.

\bibitem[WCT00a]{WCT00PROCTSROCT}
B.~Y. Wu, K.-M. Chao, and C.~Y. Tang, \emph{Approximation algorithms for some
  optimum communication spanning tree problems}, Discrete Applied Mathematics
  \textbf{102} (2000), 245--266.

\bibitem[WCT00b]{WCT00MRCT}
\bysame, \emph{Approximation algorithms for the shortest total path length
  spanning tree problem}, Discrete Applied Mathematics \textbf{105} (2000),
  273--289.

\bibitem[WCT00c]{WCT00PROCTPTAS}
\bysame, \emph{A polynomial time approximation scheme for optimal
  product-requirement communication spanning trees}, Journal of Algorithms
  \textbf{36} (2000), 182--204.

\bibitem[Wes01]{Wes01}
D.~B. West, \emph{Introduction to graph theory}, 2nd ed., Prentice-Hall, 2001.

\bibitem[Wig94]{Wig94}
A.~Wigderson, \emph{\uppercase{${\cal NL}/\text{\lowercase{poly}}\subseteq
  \oplus{\cal L}/\text{\lowercase{poly}}$}}, Proceedings of the 9th IEEE
  Structure in Complexity Conference, 1994, pp.~59--62.

\bibitem[WLB{\etalchar{+}}00]{WLBCRT00}
B.~Y. Wu, G.~Lancia, V.~Bafna, K.-M. Chao, R.~Ravi, and C.~Y. Tang, \emph{A
  polynomial time approximation scheme for minimum routing cost spanning
  trees}, SIAM Journal on Computing \textbf{29} (2000), 761--778.

\bibitem[Won80]{Won80}
R.~T. Wong, \emph{Worst-case analysis of network design problem heuristics},
  SIAM Journal on Matrix Analysis and Applications \textbf{1} (1980), 51--63.

\bibitem[Wu02]{Wu02}
B.~Y. Wu, \emph{A polynomial time approximation scheme for the two-source
  minimum routing cost spanning trees}, Journal of Algorithms \textbf{44}
  (2002), 359--378.

\end{thebibliography}
\noindent

\end{document}